\documentclass{daj}

\usepackage{amsmath,amsfonts,amssymb,amsthm,eucal,color,bbm,verbatim,graphicx}

\newtheorem{thm}{Theorem}
\newtheorem{cor}[thm]{Corollary}
\newtheorem{lemma}[thm]{Lemma}
\newtheorem{prop}[thm]{Proposition}

\newtheorem{conj}[thm]{Conjecture}
\theoremstyle{definition}
\newtheorem*{rmk}{Remark}

\newcommand{\R}{\mathbb{R}}
\newcommand{\E}{\mathbb{E}}
\newcommand{\Prob}{\mathbb{P}}
\newcommand{\N}{\mathbb{N}}

\newcommand{\Z}{\mathbb{Z}}
\newcommand{\C}{\mathbb{C}}

\newcommand{\abs}[1]{\left\vert #1 \right\vert}
\newcommand{\norm}[1]{\left\Vert #1 \right\Vert}

\newcommand{\dfn}[1]{\textbf{#1}}
\newcommand{\eps}{\varepsilon}

\DeclareMathOperator{\var}{Var}

\newcommand{\Set}[2]{\left\{#1 \mathrel{} \middle| \mathrel{} #2
  \right\}}
\DeclareMathOperator{\diam}{diam}
\newcommand{\ind}[1]{\mathbbm{1}_{#1}}

\dajAUTHORdetails{%
  title = {Self-Similarity in the Circular Unitary Ensemble}, %% please capitalize all significant words
  author = {Elizabeth S.\ Meckes and Mark W.\ Meckes},
  plaintextauthor = {Elizabeth S. Meckes, Mark W. Meckes},
  keywords = {random unitary matrices, Haar measure, eigenvalues,
    self-similarity, determinantal point processes},
}   %%% END \dajAUTHORdetails

\dajEDITORdetails{%
   year={2016},
   number={9},
   received={20 October 2015},   % received date: example: 7 January 2017
   revised={10 June 2016},    % Optional revised date (you may comment it out)
   published={15 June 2016},  % published date
   doi={10.19086/da.736},       % XXX = number of paper, e.g. da006 for paper#6
}   %%% END \dajEDITORdetails

\begin{document}

\begin{frontmatter}[classification=text]
%% EDITOR: this will force the keywords to appear right after the Abstract.
%%   If the abstract is too long and would force the keywords off the
%%   front page, please comment out % [classification=text] above
%%   This way the keywords will be floated on the bottom of the first page
%%   even though the Abstract spills over to the next page.

%%% AUTHOR: Title goes here.  This line is optional.  You must use it
%%   if title has footnote attached or requires nontrivial typesetting,
%%   e.g., inclusion of linebreaks to force nice layout.
%\title{Short Proof of R\"odl's $n^{\log\log n}$ Bound\footnote{This is a footnote to the title}} %% please capitalize all significant words

%%% AUTHOR:
%%% List all authors. If you wish, place grant acknowledgements in \thanks.
%%% In brackets include a short tag for each author.
\author[elizabeth]{Elizabeth S.\ Meckes%\thanks{Supported by...}
%Left support in separate Acknowledgements section - MM
}
\author[mark]{Mark W.\ Meckes%\thanks{Supported by...}
%Left support in separate Acknowledgements section - MM
}
%\author[laci]{L\'aszl\'o Lov\'asz\thanks{Supported by...}}
%\author[andy]{Andrew Chi-Chih Yao\thanks{Supported by...}}

%%% AUTHOR: Abstract goes here
\begin{abstract}
  This paper gives a rigorous proof of a conjectured statistical
  self-similarity property of the eigenvalues random matrices from the
  Circular Unitary Ensemble.  We consider on the one hand the
  eigenvalues of an $n \times n$ CUE matrix, and on the other hand
  those eigenvalues $e^{i\phi}$ of an $mn \times mn$ CUE matrix with
  $\abs{\phi} \le \pi / m$, rescaled to fill the unit circle.  We show
  that for a large range of mesoscopic scales, these collections of
  points are statistically indistinguishable for large $n$.  The proof
  is based on a comparison theorem for determinantal point processes
  which may be of independent interest.
\end{abstract}
\end{frontmatter}

%%% AUTHOR: body of paper starts here

%Moved remaining template material after \end{document} - MM

\section{Introduction}

The set of $N\times N$ unitary matrices is a compact Lie group, and as
such, possesses a unique probability measure which is invariant under left-
and right-translation (called Haar measure).  In random matrix theory, the unitary group
together with Haar probability measure is called the circular unitary
ensemble (CUE).  The word \emph{circular} refers to the fact that all of the
eigenvalues of a CUE matrix lie on the unit circle in the complex plane.

There has long been a folklore conjecture that the distribution of the
eigenvalues of a CUE random matrix has a self-similar structure.  For
example, in their
statistical analysis \cite{CD03} of CUE eigenvalues and zeroes of the
Riemann zeta function, Coram and Diaconis hypothesized that the
following may hold:
\begin{conj}\label{C:Coram-Diaconis}
  Let $U$ be an $N\times N$ random matrix from the CUE with
  eigenvalues $\{e^{i\theta_j}\}_{1\le j\le N}$, where $0 \le \theta_1
  \le \dots \le \theta_N < 2\pi$.  Choose an eigenvalue
  $e^{i\theta_K}$ uniformly, and let $T$ be the length of the
  counter-clockwise circular
  arc from $\theta_K$ to $\theta_{K+k}$, 
  where the indices are interpreted modulo $N$.  Let $\phi\in[0,2\pi)$
  be a uniformly chosen random angle, independent of $U$.  If $k$ and
  $N$ are both large, then the random set of points
  \[
  \left\{e^{i\left(\phi+\frac{2\pi\theta_j}{T}\right)}\right\}_{K\le
    j<K+k}
  \]
  is statistically indistinguishable from the eigenvalues of a
  $k\times k$ random matrix from the CUE.
\end{conj}

That is, a random choice of $k$ sequential eigenvalues of an $N\times
N$ CUE matrix $U$,
rescaled and randomly rotated, is indistinguishable from the full set
of eigenvalues of a $k\times k$ random matrix.

Aside from statistical evidence for the conjecture, there is a result
of E.\ Rains \cite{Ra03} which is suggestive of this kind of
self-similarity.  Suppose that $U$ is an $N\times N$ random CUE
matrix, with $N=nk$; Rains proved that the distribution of the
eigenvalues of $U^n$ is \emph{exactly} that of the collection of
eigenvalues of $n$ independent $k\times k$ random CUE matrices.  That
is, wrapping the eigenvalues of $U$ around the circle $n$ times
produces $n$ independent copies of the $k$ eigenvalues of a
$k\times k$ random matrix.  It is tempting to view each of those
collections of $k$ eigenvalues as coming from one of the $n$ arcs of
the circle that gets stretched to cover the circle once (this is not
at all the way Rains' theorem is actually proved).  If this intuition
were correct, it would illustrate exactly the kind of self-similarity
conjectured by Coram and Diaconis.

\bigskip

In this paper, we give a rigorous proof of a version of the
self-similarity conjecture.  The following notation is used
throughout.  Let $U$ be an $n\times n$ random CUE matrix with
eigenvalues $\{e^{i\theta_j}\}_{1\le j\le n}$, with
$\theta_j \in [-\pi,\pi)$ for each $j$. (It is a matter of technical
convenience to take the arguments of the eigenvalues to be in
$[-\pi,\pi)$ here instead of in $[0,2\pi)$ as in Conjecture
\ref{C:Coram-Diaconis}.) For $A \subseteq [-\pi,\pi)$,
$\mathcal{N}_{n,A}$ denotes the number of eigenangles $\theta_j$ which
lie in $A$; we generally omit the $n$ and write $\mathcal{N}_A$.  For
$\theta \in [0,\pi)$, $\mathcal{N}_{[-\theta,\theta]}$ is denoted by
$\mathcal{N}_\theta$.  For $m\ge 1$, let $U^{(m)}$ be an $nm\times nm$
random CUE matrix with eigenvalues $\{e^{i\phi_j}\}_{1\le j\le nm}$,
with $\phi_j \in [-\pi,\pi)$ for each $j$, and let
\[
\mathcal{N}_{n,A}^{(m)} = \mathcal{N}_A^{(m)} := \#\Set{j}{\phi_j \in
  \left[-\frac{\pi}{m},\frac{\pi}{m}\right),\ m\phi_j \in A};
\]
$\mathcal{N}_A^{(m)}$ counts the random points in $A$ of the point
process consisting of the eigenvalues of $U^{(m)}$ in the arc of
length $\frac{2\pi}{m}$ about $1$, and rescaling to fill out the whole
circle.  While the total number of eigenvalues in this arc is random,
it concentrates strongly at its expected value of $n$.  In the context
of the Diaconis--Coram conjecture, our $nm$ plays the role of $N$ and
$n$ plays the role of $k$.
\begin{thm}
  \label{T:TV-bound}
  Suppose that $m, n \ge 1$, and that $A \subseteq [-\pi, \pi)$ has
  diameter $\diam A \le \pi$.  Then
  \[
  d_{TV} \left( \mathcal{N}_A, \mathcal{N}^{(m)}_A \right) 
  \le W_1 \left( \mathcal{N}_A, \mathcal{N}^{(m)}_A \right) 
  \le \frac{\sqrt{mn} \abs{A} \diam A}{6\pi},
  \]
  where $d_{TV}(\cdot,\cdot)$ denotes total variation distance between
  random variables, $W_1(\cdot,\cdot)$ denotes $L^1$-Wasserstein
  distance, and $\abs{A}$ denotes the Lebesgue measure of $A$.
\end{thm}

For context, recall that $\E\mathcal{N}_A=\frac{n\abs{A}}{2\pi}$; the
same is true for $\mathcal{N}_A^{(m)}$.

In the statement of Theorem \ref{T:TV-bound}, and all of the following
results, precise constants are included for concreteness, with no
claims as to their sharpness.  The definitions of
$d_{TV}(\cdot,\cdot)$ and $W_1(\cdot,\cdot)$ are recalled at the end
of this section.

Theorem \ref{T:TV-bound} was stated with the implicit assumption that
$m$ is an integer, since it is in that case that it relates directly
to Conjecture \ref{C:Coram-Diaconis}.  However, it is only strictly
necessary that $mn$ is an integer, and a slight refinement of the
proof shows that for any $m \ge 1$, if $mn \in \N$, then
\begin{equation}
  \label{E:TV-bound-improved}
  d_{TV} \left( \mathcal{N}_A, \mathcal{N}^{(m)}_A \right) 
  \le W_1 \left( \mathcal{N}_A, \mathcal{N}^{(m)}_A \right) 
  \le C \left(1-\frac{1}{m^2}\right)\sqrt{mn} \abs{A} \diam A.
\end{equation}
In particular, this yields the comparison
\[
  d_{TV} \left( \mathcal{N}_{n,A}, \mathcal{N}_{n+1, \frac{n}{n+1}A} \right) 
  \le W_1 \left( \mathcal{N}_{n,A}, \mathcal{N}_{n+1, \frac{n}{n+1}A}\right) 
  \le C \frac{\abs{A} \diam A}{\sqrt{n}}
\]
between $n \times n$ CUE eigenvalues and $(n+1)\times (n+1)$ CUE
eigenvalues.

\medskip

As a consequence of Theorem \ref{T:TV-bound},  
if $\{A_n\}$ is a sequence of sets such that either $\diam A_n =
o(n^{-1/4})$ or $\abs{A_n} = o(n^{-1/2})$ as $n \to \infty$, then
\[
d_{TV} \left( \mathcal{N}_{A_n}, \mathcal{N}^{(m)}_{A_n} \right), 
W_1 \left(\mathcal{N}_{A_n}, \mathcal{N}^{(m)}_{A_n} \right) \to 0.
\]
Thus indeed, a sequential arc of about $n$ of the $nm$ eigenvalues of
an $nm\times nm$ random matrix is statistically indistinguishable, on
the scale of $o(n^{-1/4})$ for diameter or $o(n^{-1/2})$ for Lebesgue
measure, from the $n$ eigenvalues of an $n\times n$ random matrix.

A remarkable feature of Theorem \ref{T:TV-bound} is that it yields
\emph{microscopic} information even at a \emph{mesoscopic} scale: if
$\frac{1}{n}\ll\diam A_n\ll\frac{1}{n^{1/4}} $, then
$\mathcal{N}_{A_n}$ and $\mathcal{N}^{(m)}_{A_n}$ both have
expectations and variances tending to infinity (as follows from Lemma
\ref{T:variance-log} below).  One would thus typically try to
understand the point processes at these scales by studying statistical
properties of the recentered and rescaled counts, rather than try to
observe individual points.  Here, we are able to make direct
point-by-point comparisons of the two point processes treated as
discrete objects, with no rescaling or continuous approximations.

The fact that we are able to compare the two point processes with no
rescaling certainly suggests that we are witnessing a true
self-similarity phenomenon which is a special feature of the structure
of the eigenvalues of CUE random matrices.  However, one should be
careful to check that the two point processes are not similar simply
because they have the same limit.  Indeed, Wieand \cite{Wieand} and
Soshnikov \cite{So00} showed that
\[
\frac{\mathcal{N}_\theta - \E \mathcal{N}_\theta}{\sqrt{\var
    \mathcal{N}_\theta}} \Rightarrow N(0,1)
\]
as $n \to \infty$ for fixed $\theta$; the same then follows for
$\mathcal{N}_\theta^{(m)}$. Figure \ref{F:simulations} gives a
convincing visual illustration that $\mathcal{N}_A$ and
$\mathcal{N}_A^{(m)}$ resemble each other more closely than either
resembles a Gaussian distribution; a rigorous proof of this fact is
given in Proposition \ref{T:counting-function-clt} below.

\begin{figure}
  \centering
  \includegraphics[width=2.4in]{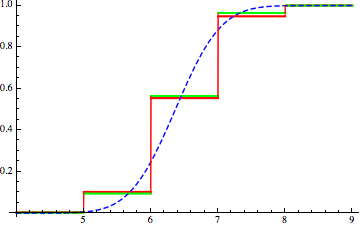}\hspace{.25in}
  \includegraphics[width=2.4in]{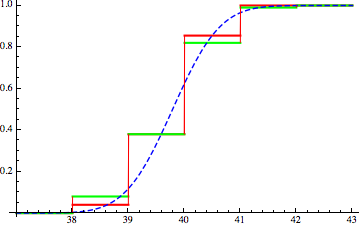}

  \begin{tabular}{|c|c|}
    \hline
    \multicolumn{2}{|c|}{Kolmogorov--Smirnov statistics} \\
    \hline
    $\mathcal{N}_{0.2}$ to Gaussian & 0.329 \\
    $\mathcal{N}^{(2)}_{0.2}$ to Gaussian & 0.319 \\
    $\mathcal{N}_{0.2}$ to $\mathcal{N}^{(2)}_{0.2}$ & 0.016 \\
    \hline
  \end{tabular}\hspace{.25in}
  \begin{tabular}{|c|c|}
    \hline
    \multicolumn{2}{|c|}{Kolmogorov--Smirnov statistics} \\
    \hline
    $\mathcal{N}_{0.25}$ to Gaussian & 0.262 \\
    $\mathcal{N}^{(2)}_{0.25}$ to Gaussian & 0.262 \\
    $\mathcal{N}_{0.25}$ to $\mathcal{N}^{(2)}_{0.25}$ & 0.04 \\
    \hline
  \end{tabular}

  \caption{ \textsc{Left:} Simulated cumulative distribution functions (500
    trials) for $\mathcal{N}_{0.2}$ (red) and
    $\mathcal{N}^{(2)}_{0.2}$ (green), with $n = 100$. 
    \newline
    \textsc{Right:} Simulated cumulative distribution functions (200 trials)
    for $\mathcal{N}_{0.25}$ (red) and $\mathcal{N}^{(2)}_{0.25}$
    (green), with $n = 500$.  
    \newline 
    The dotted lines show
    Gaussian cumulative distribution functions with mean equal to the
    theoretical mean of both $\mathcal{N}_\theta$ and
    $\mathcal{N}^{(2)}_\theta$ (i.e., $\frac{20}{\pi} \approx 6.4$ and
    $\frac{125}{\pi} \approx 39.8$, respectively) and variance equal
    to the average of the two corresponding sample variances.}
  \label{F:simulations}
\end{figure}

We conjecture that a comparable result to Theorem \ref{T:TV-bound}
holds without the restriction on $\diam A$, and that the factor of
$\sqrt{n}$ in the right hand side is an artifact of our proof; this
would imply that $\mathcal{N}_{A_n}$ and $\mathcal{N}_{A_n}^{(m)}$
become indistinguishable as long as $\abs{A_n} \to 0$.  For more
details, see the remark at the end of Section \ref{S:dpp}.  On the
other hand, we do not expect such a result to hold for sets of
constant size; i.e., independent of $n$.  For example, Rains
\cite{Ra97} gives precise asymptotics for $\var \mathcal{N}_\theta$
for $n \to \infty$ and $\theta$ fixed, which show that
$\var \mathcal{N}_\theta$ and $\var \mathcal{N}_\theta^{(m)}$ are not
asymptotically equal.  This suggests (but does not formally imply)
that Theorem \ref{T:TV-bound} does not hold in this setting.  Rains'
estimate does show that Proposition \ref{T:variances-almost-equal}
below on the asymptotic equality of variances does not extend to that
regime.

%have mean and 
%fluctuations growing with $n$ (see Proposition \ref{T:variance-log} below)

We expect that a version of Theorem \ref{T:TV-bound} holds for the
other circular ensembles of random matrix theory; however, our
approach is via the determinantal structure of the eigenvalue process
for the CUE, which is not present outside the unitary case.

Finally, some comments on the relationship between Conjecture
\ref{C:Coram-Diaconis} and Theorem \ref{T:TV-bound} are in order.  The
models of self-similarity being used are not identical; in Conjecture
\ref{C:Coram-Diaconis}, exactly $k+1$ sequential eigenvalues are
selected and stretched as needed to make the first and last meet,
resulting in exactly $k$ random points.  In Theorem \ref{T:TV-bound},
the eigenvalues from an arc making up a fixed fraction of the circle
are chosen and that arc is stretched (deterministically) to cover the
whole circle; the resulting total number of points is random.
However, in the mesoscopic regime the two models are essentially the
same.  The idea is the following: eigenvalue rigidity (see Lemma 10 of
\cite{MM-powers}) implies that the difference between the
$j^\mathrm{th}$ and the $(j+n)^\mathrm{th}$ eigenangles of an
$nm\times nm$ CUE matrix is about
$\frac{2\pi}{m} + O\bigl(\frac{\sqrt{\log n}}{n}\bigr)$ with high
probability. So whereas Theorem \ref{T:TV-bound} considers the
eigenangles of an $nm\times nm$ matrix in an interval of length
$\theta / m$, Conjecture \ref{C:Coram-Diaconis} suggests considering
the eigenangles in an interval whose length is random but typically
about
$\frac{\theta}{m} + \theta O\bigl(\frac{\sqrt{\log n}}{n}\bigr)$.  But
if $\theta \ll \frac{1}{\sqrt{\log n}}$, then with extremely high
probability an interval of length $\theta \frac{\sqrt{\log n}}{n}$
contains no eigenangles, and so the corresponding counts are the same.

\bigskip

The rest of this paper is organized as follows.  In Section
\ref{S:dpp} we give the background and general results on
determinantal point processes needed to prove Theorem
\ref{T:TV-bound}, followed by the proof of the theorem and a corollary
giving a rate for the classical convergence of the eigenvalue process
to the sine kernel process on a microscopic scale.  In Section
\ref{S:variances} we give precise asymptotics for the variances of the
counting functions.  As a consequence, we are able to identify a sharp
rate of convergence in the central limit theorem mentioned above,
which is in particular much slower than the merging of distributions
in Theorem \ref{T:TV-bound}.  We also show that the variances of the
counting functions of the two processes are asymptotically equal
throughout the entire mesoscopic regime, giving a rigorous proof of
another manifestation of the self-similarity phenomenon.  Finally,
Section \ref{S:correlations} gives a surprising comparison between the
joint intensities of the eigenvalues processes for $U$ and $U^{(m)}$.

We conclude this section with a brief review of the notions of
distance used here.  The following distances can be defined much more
generally, but for our purposes, it suffices to define them for
integer-valued random variables $X$ and $Y$.
\begin{enumerate}
\item The \dfn{total variation distance} from $X$ to $Y$ is defined by
  \begin{equation*}%\label{E:tv}
    d_{TV}(X,Y):=\sup_{A\subseteq\Z}\abs{\Prob[X\in A]-\Prob[Y\in A]}.
  \end{equation*}
%  Equivalently, 
%\[d_{TV}(X,Y)=\frac{1}{2}\sum_{k\in\Z}\abs{\Prob[X=k]-\Prob[Y=k]}.\] 
\item The \dfn{$L^1$-Wasserstein distance} is defined by
  \begin{equation*}%\label{E:W_p}
    W_1(X,Y) :=
    \inf_{(Z_1,Z_2)}\E\abs{Z_1-Z_2},
  \end{equation*}
  where the infimum is over random vectors $(Z_1,Z_2)$ such that $Z_1$
  has the same distribution as $X$ and $Z_2$ has the same distribution
  as $Y$ (such a random vector is called a coupling of $X$ and $Y$).

  The Kantorovich--Rubenstein Theorem states that $W_1$ can
  equivalently be defined as
  \[
  W_1(X,Y):=\sup_f\abs{\E f(X)-\E f(Y)},
  \]
  where the supremum is over 1-Lipschitz functions $f:\Z\to\R$.  The
  distance $W_1$ is a metric for the topology of weak convergence plus
  convergence of absolute first moments. (See \cite[Section
  6]{Villani} for a thorough discussion and proofs.)
\end{enumerate}

Note that an indicator function of a set $A$ of integers is
1-Lipschitz on $\Z$, and so for $X$ and $Y$ integer-valued,
\begin{equation}
  \label{E:distance-comparison}
  d_{TV}(X,Y)\le W_1(X,Y).
\end{equation}

\section{Determinantal point processes and the proof of Theorem
  \ref{T:TV-bound}}\label{S:dpp}

Let $\Lambda$ be a locally compact Polish space.  A simple point
process on $\Lambda$ is a random integer-valued (positive) Radon
measure $\chi$ on $\Lambda$, such that the measure of any singleton is
at most $1$.  Alternatively, it may be viewed as a locally finite
random set of points in $\Lambda$; if $A \subseteq \Lambda$ then we
write $\mathcal{N}_A = \chi(A)$ for the (random) number of points
lying in $A$.  If $\Lambda$ is equipped with a reference Borel measure
$\mu$, then the $k^{\mathrm{th}}$ \dfn{joint intensity} or
\dfn{correlation function} $\rho_k:\Lambda^k \to [0,\infty)$ of $\chi$
is defined by the equation
\[
\E \left[\prod_{i=1}^k \mathcal{N}_{A_i}\right] = \int_{A_1} \dots \int_{A_k}
\rho_k(x_1, \dots, x_k) \ d\mu(x_1) \dots d\mu(x_k),
\]
whenever $A_1, \dots, A_k \subseteq \Lambda$ are measurable and
pairwise disjoint, assuming that such functions exist.  A simple point
process is called a \dfn{determinantal point process} with kernel
$K: \Lambda^2 \to \C$ if its joint intensities exist and
\[
\rho_k(x_1, \dots, x_k) = \det \left[K(x_i,x_j)\right]_{i,j=1}^k.
\]
Note that it is immediate from the definition that the restriction of
a determinantal point process on $\Lambda$ to a measurable subset
$D\subseteq\Lambda$ is again a determinantal point process.

A kernel $K:\Lambda^2\to\C$ defines an integral operator on $L^2(\mu)$
by 
\begin{equation}
  \label{E:integral-operator}
  \mathcal{K}(f)(x):=\int_\Lambda K(x,y)f(y)d\mu(y);
\end{equation}
if $K(x,y)=\overline{K(y,x)}$, then the operator $\mathcal{K}$ is
self-adjoint.  It was proved by Macchi \cite{Macchi} and Soshnikov
\cite{Soshnikov} that a kernel $K$ which defines a self-adjoint, trace
class operator $\mathcal{K}$ as above is the kernel of a determinantal
point process if and only if all of the eigenvalues of $\mathcal{K}$
lie in $[0,1]$.

For the remainder of this paper, $\chi$ will denote the point process
of eigenvalue angles in $[-\pi,\pi)$ of an $n\times n$ CUE random
matrix.  For fixed $m \ge 1$, let $\chi^{(m)}$ denote the point
process obtained by multiplying by $m$ those eigenvalue angles of an
$nm\times nm$ CUE random matrix which lie in
$\bigl[-\frac{\pi}{m}, \frac{\pi}{m}\bigr)$.

It is a fact originally due to Dyson that $\chi$ is a
determinantal point process on $[-\pi,\pi)$; it follows easily that
$\chi^{(m)}$ is as well.  The following Proposition gives explicit
formulae for the corresponding kernels.

\begin{prop}
  \label{T:kernels}
  The point process $\chi^{(m)}$ on $[0,2\pi)$ is determinantal with
  kernel
  \[
  K_n^{(m)}(x,y) = \frac{1}{2\pi}
  \frac{\sin\left(\frac{n(x-y)}{2}\right)}{m\sin\left(\frac{(x-y)}{2m}\right)}.
  \]
  with respect to Lebesgue measure.
\end{prop}

\begin{proof}
  The case $m=1$ was proved by Dyson in \cite{Dyson} (although that work
  predates the language of determinantal point processes); see also
  \cite[Section 11.1]{Mehta} or \cite[Section 5.4]{KaSa}.
  The general case follows from a change of variables which shows that
  $K_n^{(m)}(x,y) = \frac{1}{m}K_{mn}^{(1)}\bigl(\frac{x}{m}, \frac{y}{m}\bigr)$.
\end{proof}

Note in particular that the corresponding operators
$\mathcal{K}_n^{(m)}$ as defined in \eqref{E:integral-operator} are
self-adjoint and trace class.

\medskip

The following general result on determinantal point processes is the
main technical ingredient behind Theorem \ref{T:TV-bound}.

\begin{prop}
  \label{T:general-bound}
  Let $\mathcal{N}$ and $\widetilde{\mathcal{N}}$ be the total numbers
  of points in two determinantal point processes on $(\Lambda, \mu)$
  with conjugate-symmetric kernels $K, \widetilde{K} \in
  L^2(\mu\otimes \mu)$, respectively. Suppose that
  $\mathcal{N}, \widetilde{\mathcal{N}} \le N$ almost surely.  Then
  \[
  d_{TV}(\mathcal{N}, \widetilde{\mathcal{N}}) 
  \le W_1(\mathcal{N}, \widetilde{\mathcal{N}})
  \le \sqrt{N \int \int \abs{K(x,y) - \widetilde{K}(x,y)}^2 \
    d\mu(x) d\mu(y)}.
  \]
\end{prop}

Proposition \ref{T:general-bound} depends on the following remarkable
property of determinantal point processes.

\begin{lemma}[{\cite[Theorem 7]{HKPV06}}]
  \label{T:HKPV}
  Consider a determinantal
  point process with kernel $K$, whose corresponding  integral
  operator $\mathcal{K}$ is self-adjoint and trace class, with
  eigenvalues $\{ \lambda_j \}$.  Let $\mathcal{N}$ be the total
  number of points in the process.  
  Then 
\[\mathcal{N} \stackrel{d}{=} \sum_j \xi_j,\]
 where $\{ \xi_j \}$
  are independent Bernoulli random variables with $\Prob[\xi_j = 1] =
  \lambda_j$ and $\Prob[\xi_j = 0] = 1-\lambda_j$.
\end{lemma}

\begin{proof}[Proof of Proposition \ref{T:general-bound}] 
  By \eqref{E:distance-comparison}, it suffices to prove the second
  inequality.
  
  Let $\{\lambda_j\}$ and $\{ \widetilde{\lambda}_j \}$ be the
  eigenvalues, listed in nonincreasing order, of the integral
  operators $\mathcal{K}$ and $\widetilde{\mathcal{K}}$ with kernels
  $K$ and $\widetilde{K}$ respectively. Since
  $\mathcal{N}, \widetilde{\mathcal{N}} \le N$, by Lemma \ref{T:HKPV},
  $\lambda_j = \widetilde{\lambda}_j = 0$ for $j > N$. Let
  $\{Y_j\}_{j=1}^N$ be independent random variables uniformly
  distributed in $[0,1]$.  For each $j$, define
  \[
  \xi_j = \ind{Y_j \le \lambda_j}
  \qquad \text{and} \qquad
  \widetilde{\xi}_j = \ind{Y_j \le \widetilde{\lambda}_j}.
  \]
  Through Lemma \ref{T:HKPV}, this gives a coupling of
  $\mathcal{N}$ and $\widetilde{\mathcal{N}}$, and so
  \begin{equation}\label{E:W1-bound}
  W_1(\mathcal{N}, \widetilde{\mathcal{N}})
  \le \E \abs{ \sum_{j=1}^N \xi_j -
    \sum_{j=1}^N \widetilde{\xi}_j }
  \le \sum_{j=1}^N \E \abs{ \xi_j - \widetilde{\xi}_j }
  = \sum_{j=1}^N \abs{\lambda_j - \widetilde{\lambda}_j}
  \le \sqrt{ N \sum_{j=1}^N \abs{\lambda_j -
          \widetilde{\lambda}_j}^2}.
  \end{equation}
  By the Hoffmann--Wielandt inequality \cite[Theorem II.6.11]{Kato95},
  \[
  \sqrt{\sum_{j=1}^N \abs{\lambda_j - \widetilde{\lambda}_j}^2}\le
  \norm{\mathcal{K}-\widetilde{\mathcal{K}}}_{H.S.},
  \]
  where $\lVert\cdot\rVert_{H.S.}$ denotes the Hilbert--Schmidt norm.
  The result now follows from the general fact that the
  Hilbert--Schmidt norm of an integral operator on $L^2(\mu)$ is given
  by the $L^2(\mu\otimes \mu)$ norm of its kernel (see e.g.\ \cite[p.\
  245]{Wojtaszczyk}).
\end{proof}

We are now in a position to prove the main theorem.
\begin{proof}[Proof of Theorem \ref{T:TV-bound}]
  For every $0 \le \varphi \le \frac{\pi}{2}$,
  \begin{equation}
    \label{E:Taylor}
  \varphi - \frac{1}{6}\varphi^3
  \le \sin \varphi \le m \sin \left(\frac{\varphi}{m} \right)
  \le \varphi,
  \end{equation}
  and so
  \begin{equation*}
      0 \le \frac{1}{\sin \varphi} - \frac{1}{m \sin
          \left(\frac{\varphi}{m}\right)} 
        \le \frac{1}{\varphi - \frac{1}{6} \varphi^3}
        - \frac{1}{\varphi} \\
        = \frac{\varphi}{6 - \varphi^2} \le \frac{\varphi}{3}.
  \end{equation*}
  Thus by Propositions \ref{T:kernels} and \ref{T:general-bound},
  \begin{equation*}
    \begin{split}
      W_1(\mathcal{N}_A, \mathcal{N}^{(m)}_A) & \le
      \sqrt{\frac{mn}{(2\pi)^2} \int_A \int_A \sin^2
        \left(\frac{n(x-y)}{2}\right) \left( \frac{1}{\sin\left(
              \frac{x-y}{2}\right)} - \frac{1}{m \sin
            \left(\frac{x-y}{2m} \right)}\right)^2 \ dx \ dy} \\
      & \le \frac{1}{6\pi} \sqrt{mn \int_A \int_A (x-y)^2 \ dx \ dy} \\
      & \le \frac{1}{6\pi} \sqrt{mn} \abs{A} \diam A.
      \qedhere
    \end{split}
  \end{equation*}
\end{proof}

The refinement \eqref{E:TV-bound-improved} of Theorem \ref{T:TV-bound}
follows by using a higher-order Taylor expansion in \eqref{E:Taylor}.

Both $\mathcal{N}_A$ and $\mathcal{N}_A^{(m)}$ satisfy central limit
theorems in the mesoscopic regime (see Proposition \ref{T:counting-function-clt} and the remark which follows).
We show in the next section that Theorem \ref{T:TV-bound} does indeed describe a
non-trivial self-similarity phenomenon on a mesoscopic level, which is
not the result of both processes having the same limit.  

In the microscopic regime, one can say more.  As was first observed in
\cite{Dyson}, and more clearly spelled out in \cite{Mehta}, the kernel
$K_n^{(1)}=K_n$ has the following microscopic scaling limit:
\begin{equation}
  \label{E:sine-kernel}
  \lim_{n\to \infty} \frac{2\pi}{n} K_n \left(\frac{2\pi x}{n},
    \frac{2\pi y}{n}\right) = \frac{\sin \bigl(\pi(x-y)\bigr)}{\pi
    (x-y)}.
\end{equation}
The same microscopic scaling limit appears for bulk eigenvalues of
certain Hermitian random matrices as well; see \cite{AGZ,CoLe,Mehta}.
There is sufficient uniformity in the convergence in
\eqref{E:sine-kernel} to imply that the point process $\chi$, rescaled
to lie in $[-n/2,n/2)$, converges as $n \to \infty$ to an unbounded
point process on $\R$ which is determinantal, with the right hand side
of \eqref{E:sine-kernel} as its kernel with respect to Lebesgue
measure.  This process is called the sine kernel process; we denote by
$\mathcal{S}_A$ the number of points of the sine kernel process which
lie in $A \subseteq \R$.  In particular, by e.g.\ \cite[Lemma
4.2.48]{AGZ}, $\mathcal{N}_{n,\frac{2\pi}{n}A} \Rightarrow
\mathcal{S}_A$.  A limited version of Theorem \ref{T:TV-bound} can be
deduced from the convergence to the sine kernel process.  On the other
hand, Theorem \ref{T:TV-bound} actually improves on the classical
microscopic result by estimating a rate of convergence, as follows.

\begin{cor}
  \label{T:sine-kernel-convergence} 
  Let $A \subseteq \R$, and let $\mathcal{S}_A$ denote the number of
  points of the sine kernel process which lie in $A$.  Then
  \[
  d_{TV}\bigl(\mathcal{N}_{\frac{2\pi}{n} A} , \mathcal{S}_A \bigr) \le W_1
  \bigl(\mathcal{N}_{\frac{2\pi}{n} A} , \mathcal{S}_A )\bigr) \le
  \frac{5 \abs{A} \diam A}{n^{3/2}}
  \]
  for all sufficiently large $n$.
\end{cor}

\begin{proof}
  Let $n$ be large enough that $A \subseteq \left[-\frac{n}{2},
    \frac{n}{2}\right)$ and $\diam A \le \frac{n}{2}$.  Let $k\ge 0$.
  Recall that by definition of $\chi^{(2)}$,
  \[
  \mathcal{N}^{(2)}_{2^k n, \frac{2\pi}{2^k n}A} = 
  \mathcal{N}_{2^{k+1} n, \frac{2\pi}{2^{k+1} n}A}.
  \]
  It thus follows from Theorem \ref{T:TV-bound} (with $m=2$ and $2^kn$
  in place of $n$) that
  \[
  W_1\left(\mathcal{N}_{2^k n, \frac{2\pi}{2^k n}A},
    \mathcal{N}_{2^{k+1} n, \frac{2\pi}{2^{k+1} n}A} \right) \le
  \frac{\sqrt{2^{k+1} n}}{6\pi} \frac{4\pi^2 \abs{A} \diam A }{(2^k
    n)^2}=\frac{2\sqrt{2} \pi \abs{A} \diam A}{3 (2^kn)^{3/2}}.
  \]
  Fixing $M\in\N$ and applying this estimate for each
  $k\in\{0,\ldots,M-1\}$ then gives that
  \begin{equation}
    \label{E:W1-sine-bound}
    \begin{split}
      W_1 \left(\mathcal{N}_{\frac{2\pi}{n}A}, \mathcal{S}_A \right) &
      \le \sum_{k=0}^{M-1} W_1 \left(\mathcal{N}_{2^k n, \frac{2\pi}{2^k
            n}A}, \mathcal{N}_{2^{k+1} n, \frac{2\pi}{2^{k+1} n}A}
      \right) + W_1 \left(\mathcal{N}_{2^{M} n, \frac{2\pi}{2^{M} n}A},
      \mathcal{S}_A \right) \\
      & \le \frac{2\sqrt{2} \pi \abs{A} \diam A}{3 n^{3/2}} \sum_{k=0}^{M-1}
      \frac{1}{2^{3k/2}}
      + W_1 \left(\mathcal{N}_{2^{M} n, \frac{2\pi}{2^{M} n}A},
      \mathcal{S}_A \right).
    \end{split}
  \end{equation}
  As was discussed above, it is well known that $\mathcal{N}_{2^M n,
    \frac{2\pi}{2^M n} A} \Rightarrow \mathcal{S}_A$ as $M \to
  \infty$.  Since all of the $\mathcal{N}_{2^M n, \frac{2\pi}{2^M n}
    A} $ and $\mathcal{S}_A$ are nonnegative random variables with
  means equal to $\abs{A}$, weak convergence is equivalent to $W_1$
  convergence, and so $W_1 \left(\mathcal{N}_{2^{M} n,
      \frac{2\pi}{2^{M} n}A}, \mathcal{S}_A \right) \to 0$ as $M \to
  \infty$.  Thus taking the limit $M \to \infty$ in
  \eqref{E:W1-sine-bound} yields
  \[
  W_1 \left(\mathcal{N}_{\frac{2\pi}{n}A}, \mathcal{S}_A \right)
  \le \frac{2\sqrt{2} \pi \abs{A} \diam A}{3 n^{3/2}} \sum_{k=0}^{\infty}
      \frac{1}{2^{3k/2}}\le\frac{5\abs{A}\diam A}{n^{3/2}}.
      \qedhere
  \]
\end{proof}

\begin{rmk}
  The application of the Cauchy--Schwarz inequality in the last step
  of \eqref{E:W1-bound} in the proof of Proposition
  \ref{T:general-bound} above is the source of the factor of
  $\sqrt{n}$ in the statement of Theorem \ref{T:TV-bound}, which we
  conjecture to be unnecessary.  A direct estimate of the quantity
  \[
  \sum_{j=1}^N\abs{\lambda_j-\widetilde{\lambda}_j},
  \]
  which is bounded by the trace class norm of the difference
  $\mathcal{K}-\widetilde{\mathcal{K}}$, could potentially avoid that
  dimensional factor, thereby increasing the size of the mesoscopic
  regime in which Theorem \ref{T:TV-bound} gives non-trivial
  information.  Unfortunately, trace class norms are considerably more
  difficult to compute than Hilbert--Schmidt norms, and we have not
  found an estimate which improves on the approach taken above.
\end{rmk}

\section{Some further asymptotics}\label{S:variances}

The following lemma gives asymptotics for $\var
\mathcal{N}_\theta^{(m)}$ in various regimes.  As was mentioned in the
introduction, the paper \cite{Ra97} gives precise asymptotics
as $n \to \infty$ for $\var \mathcal{N}_\theta$ when $\theta$ is
fixed, but in the present context,  estimates for when $\theta$ varies with
$n$ are needed.   
\begin{lemma}
  \label{T:variance-log}
  Let $m\in\N$ be fixed.  Whenever $\frac{3 \pi}{2n} \le \theta \le \frac{\pi}{2}$,
  \[
  \var \mathcal{N}_\theta^{(m)} \ge \frac{1}{3\pi^2} \log
  \left(\frac{2 n\theta}{3\pi}\right).
  \]
  Moreover,
  \[
  \var \mathcal{N}_\theta^{(m)} \le
  \begin{cases} \frac{n^2 \theta^2 + 2}{4} & \text{if } 0 < \theta
    \le \frac{1}{n}, \\
    \frac{1}{2} \log \bigl(e^{3/2} n\theta \bigr)
    & \text{if } \frac{1}{n} \le \theta \le \frac{\pi}{2}. \end{cases}
  \]
  Consequently, for a sequence
  $\left\{\theta_n \in \left(0,\frac{\pi}{2}\right]\right\}$,
  \[
  \var \mathcal{N}_{n, \theta_n}^{(m)} \to \infty
  \qquad \text{if and only if} \qquad
  n \theta_n \to \infty.
  \]
\end{lemma}

\begin{proof}
  Observe that $\mathcal{N}_{n,[-\theta,\theta]}^{(m)}$ has the same
  distribution as
  $\mathcal{N}^{(1)}_{mn,[-\theta/m,\theta/m]} =
  \mathcal{N}_{mn,[-\theta/m,\theta/m]}$;
  since $m$ is fixed (i.e., independent of $n$), it therefore suffices
  to prove the estimates for $m=1$.  Using general formulae for
  determinantal point processes, it is shown in the proof of
  \cite[Proposition 8]{MM-powers} that
  \begin{equation}
    \label{E:variance-formula}
    \var \mathcal{N}_\theta
    = \frac{1}{2\pi^2} \left[\int_0^{2\theta}
      \frac{z\sin^2\left(\frac{nz}{2}\right)}
      {\sin^2\left(\frac{z}{2}\right)} \ dz
      + 2\theta \int_{2 \theta}^{\pi} \frac{\sin^2\left(\frac{nz}{2}\right)}
      {\sin^2\left(\frac{z}{2}\right)} \ dz \right].
  \end{equation}

  Now suppose that $\frac{3\pi}{2n} \le \theta \le \frac{\pi}{2}$. If $z \ge
  \eps \ge \frac{2\pi}{n}$, then $z - \frac{\pi}{n} \ge \frac{z}{2}$,
  and so
  \begin{align*}
      \int_\eps^{2\theta} \frac{\sin^2 \left(\frac{nz}{2}\right)}{z} \ dz
      & = \int_\eps^{2\theta} \frac{1 - \cos^2
        \left(\frac{nz}{2}\right)}{z}
      \ dz \\
      & = \log \left(\frac{2\theta}{\eps}\right) - \int_\eps^{2\theta}
      \frac{\sin^2 \left(\frac{nz}{2} + \frac{\pi}{2} \right)}{z}
      \ dz \\
      & = \log \left(\frac{2\theta}{\eps}\right) - \int_{\eps+
        \frac{\pi}{n}}^{2\theta + \frac{\pi}{n}}
      \frac{\sin^2 \left(\frac{nz}{2} \right)}{z - \frac{\pi}{n}}
      \ dz \\
      & \ge \log \left(\frac{2\theta}{\eps}\right) -
      \int_{2\theta}^{2\theta+ \frac{\pi}{n}} 
      \frac{\sin^2 \left(\frac{nz}{2} \right)}{z - \frac{\pi}{n}}
      \ dz
      - \int_\eps^{2\theta} 
      \frac{\sin^2 \left(\frac{nz}{2} \right)}{z - \frac{\pi}{n}}
      \ dz \\
      & \ge \log \left(\frac{2\theta}{\eps}\right) -
      \int_{2\theta}^{2\theta+ \frac{\pi}{n}} 
      \frac{1}{z - \frac{\pi}{n}}
      \ dz
      - 2 \int_\eps^{2\theta} 
      \frac{\sin^2 \left(\frac{nz}{2} \right)}{z}
      \ dz \\
      & = \log \left(\frac{2\theta - \frac{\pi}{n}}{\eps}\right)
      - 2 \int_\eps^{2\theta} 
      \frac{\sin^2 \left(\frac{nz}{2} \right)}{z}
      \ dz.
  \end{align*}
  Thus
  \[
  \int_\eps^{2\theta}
  \frac{\sin^2 \left(\frac{nz}{2}\right)}{z} \ dz \ge \frac{1}{3} \log
  \left(\frac{2\theta - \frac{\pi}{n}}{\eps}\right).
  \]
  Now setting $\eps = \sqrt{\frac{2\pi}{n}(2\theta - \frac{\pi}{n})}
  \ge \frac{2\pi}{n}$,
  \begin{equation*}
      \log \left(\frac{2\theta - \frac{\pi}{n}}{\eps}\right) 
      = \frac{1}{2} \log \left[\frac{n}{2\pi} \left(2\theta -
          \frac{\pi}{n}\right) \right] 
      \ge \frac{1}{2} \log \left(\frac{2n \theta}{3\pi}\right),
  \end{equation*}
  and so by \eqref{E:variance-formula},
  \[
  \var \mathcal{N}_\theta \ge \frac{1}{2\pi^2} \int_0^{2\pi}
  \frac{z\sin^2\left(\frac{nz}{2}\right)}
  {\sin^2\left(\frac{z}{2}\right)} \ dz \ge \frac{2}{\pi^2}
  \int_0^{2\pi} \frac{\sin^2\left(\frac{nz}{2}\right)} {z} \ dz \ge
  \frac{1}{3\pi^2} \log \left(\frac{2n \theta}{3\pi}\right).
  \]

  For the upper bound, observe that $\sin \left(\frac{z}{2}\right) \ge
  \frac {z}{\pi}$ for $0 \le z \le \pi$, thus
  \begin{equation*}
    \begin{split}
      2 \theta \int_{2\theta}^\pi \frac{\sin^2 \left( \frac{nz}{2}
        \right)}{\sin^2 \left( \frac{z}{2} \right)} \ dz & \le 2\theta
      \int_{2\theta}^\pi \frac{\pi^2}{z^2} \ dz \le \pi^2,
    \end{split}
  \end{equation*}
  and
  \begin{equation*}
    \begin{split}
      \int_0^{2\theta} \frac{z \sin^2 \left( \frac{nz}{2} \right)}{\sin^2
        \left( \frac{z}{2} \right)} \ dz & \le \int_0^{2\theta}
      \frac{\pi^2 n^2 z}{4} \ dz = \frac{\pi^2 n^2 \theta^2}{2}
    \end{split}
  \end{equation*}
  for any $\theta$, while if $\frac{1}{n} \le \theta \le \pi$, then
  \begin{equation*}
    \begin{split}
      \int_0^{2\theta} \frac{z \sin^2 \left( \frac{nz}{2} \right)}{\sin^2
        \left( \frac{z}{2} \right)} \ dz & \le \int_0^{2/n}
      \frac{\pi^2 n^2 z}{4} \ dz 
      + \int_{2/n}^{2\theta} \frac{\pi^2}{z} \ dz
      = \pi^2 \left[ \frac{1}{2} + \log
        (n\theta)\right].
      \qedhere
    \end{split}
  \end{equation*}
\end{proof}

\medskip

One of the consequences of the lemma is that it allows us to identify
the regime in which the (centered, normalized) counting function has a
Gaussian limit, and to provide the estimates of the rate of
convergence to Gaussian in that regime given in Proposition
\ref{T:counting-function-clt} below.  The real point of the
proposition is that the convergence of the centered, normalized
counting functions of either point process to a Gaussian limit is much
slower than the merging of distributions given in Theorem
\ref{T:TV-bound}, meaning that the resemblance between $\mathcal{N}_A$
and $\mathcal{N}_A^{(m)}$ is emphatically \emph{not} a consequence of
the central limit theorem.

\begin{prop}
  \label{T:counting-function-clt}
  For $0 \le \theta \le \frac{\pi}{2}$ and $n \ge 1$, define
  \[
  X_{n,\theta} := \frac{\mathcal{N}_{n,\theta} -
    \frac{n\theta}{\pi}}{\sqrt{\var \mathcal{N}_{n,\theta}}}.
  \]
  For each $n$, let $\theta_n\in\left[0,\frac{\pi}{2}\right]$.  The
  sequence $\{X_{n,\theta_n}\}$ converges weakly
  to the standard Gaussian
  distribution as $n \to \infty$ if and only if $n\theta_n \to
  \infty$.
   Moreover, whenever $\frac{3\pi}{n} \le \theta \le \pi$,
  \[
  \frac{3\sqrt{2}}{32\sqrt{\log \bigl(e^{3/2} n \theta \bigr)}} \le \sup_{t \in
    \R} \abs{ \Prob\left[ X_{n,\theta} \le t \right] -
    \frac{1}{\sqrt{2\pi}} \int_{-\infty}^t e^{-x^2/2} \ dx } \le
  \frac{3 \sqrt{3} \pi}{\sqrt{\log
      \left(\frac{2n\theta}{3\pi}\right)}}.
  \]
\end{prop}

\begin{rmk}
For $m\in\N$ fixed, a central limit theorem for
$\mathcal{N}_{n,\theta}^{(m)}=\mathcal{N}_{nm,\frac{\theta}{m}}$ follows immediately
from Proposition \ref{T:counting-function-clt}.
\end{rmk}
\smallskip

\begin{proof}
  First observe that for any integer-valued random variable $X$ with
  finite second moment, it follows from  Chebychev's inequality that
  \begin{equation*}\begin{split}
      \frac{3}{4} \le \Prob \left[ \abs{X - \E X} < 2 \sqrt{\var
          X}\right] 
      &= \sum_{\substack{k \in \Z, \\ \abs{k - \E X} < 2 \sqrt{\var X}}}
      \Prob[X = k] 
      \le \max_k \Prob[X=k] \left(4 \sqrt{ \var X}\right).
    \end{split} 
  \end{equation*} 
  The cumulative distribution function of $X$ thus has a jump of at
  least $\frac{3}{16\sqrt{\var X}}$ at some integer, and so
  \[
  \sup_{t\in \R} \abs{ \Prob [X \le t] - \Prob [Y \le t]} \ge
    \frac{3}{32 \sqrt{\var X}}
  \]
  for any continuous random variable $Y$.  Now,
  \[
  \sup_{t \in \R} \abs{ \Prob\left[ X_{n,\theta} \le t \right] -
    \frac{1}{\sqrt{2\pi}} \int_{-\infty}^t e^{-x^2/2} \ dx }
  =\sup_{t\in\R}\Big|\Prob\left[\mathcal{N}_\theta\le
    t\right]-\Prob\left[Y\le t\right]\Big|,
  \]
  where $Y$ is a Gaussian random variable with the same mean and
  variance as $\mathcal{N}_\theta$.  Since $\mathcal{N}_\theta$ is
  integer-valued, together with Lemma \ref{T:variance-log} this proves
  both the lower bound in the proposition and the fact that
  $X_{n, \theta_n}$ can only have a Gaussian limit if
  $n\theta_n \to \infty$.

  For the other estimate, the Berry--Esseen theorem (see, e.g.,
  \cite[Theorem XVI.5.1]{Feller2}) implies that if $\{Y_i\}_{i=1}^n$
  are independent random variables in $[0,1]$ and $X = \sum_{i=1}^n
  Y_i$, then
  \begin{equation*}
    \begin{split}
      \sup_{t\in\R} \abs{\Prob\left[\frac{X - \E X}{\sqrt{\var
              X}}\right] - \frac{1}{\sqrt{2\pi}} \int_{-\infty}^t
          e^{-x^2} \ dx }
      & \le \frac{3}{(\var X)^{3/2}} \sum_{i=1}^n \E \abs{Y_i -
        \E Y_i}^3 \\
      & \le \frac{3}{(\var X)^{3/2}} \sum_{i=1}^n \E (Y_i - \E Y_i)^2 \\
      & = \frac{3}{\sqrt{\var X}}.
    \end{split}
  \end{equation*}
  By Lemma \ref{T:HKPV}, this may be applied to
  $X_{n,\theta}$, and so Lemma \ref{T:variance-log} implies the upper
  bound in the proposition.
\end{proof}

As discussed in the introduction, we conjecture that Theorem
\ref{T:TV-bound} holds for any shrinking sequence of sets
$A_n \subseteq [-\pi, \pi)$.  We are not able to prove full
distributional comparisons for the entire regime; however the
following result shows that equality of means and asymptotic equality
of variances does hold throughout the entire mesoscopic regime. That
is, if $\{A_n\}$ is any sequence of subsets of $[-\pi,\pi)$ such that
$\diam A_n \le \pi$ eventually and $\abs{A_n} \to 0$ (in particular,
if $\diam A_n \to 0$), then
\[
\abs{\var \mathcal{N}_{A_n} - \var \mathcal{N}^{(m)}_{A_n}} \to 0
\]
as $n \to \infty$.  For context, recall that it has already been shown
that $\var\mathcal{N}_{\theta}$ itself, and thus
$\var\mathcal{N}_\theta^{(m)}$ as well, is of order $\log(n\theta)$
when $\theta \ge \frac{1}{n}$.

\begin{thm}
  \label{T:variances-almost-equal}
  For each $m , n \ge 1$ and $A \subseteq [-\pi, \pi)$,
  \[
  \E\mathcal{N}_A=\E\mathcal{N}_A^{(m)}.
  \] 
  If in addition $\diam A \le \pi$, then
  \[
  0 \le \var \mathcal{N}_A - \var
  \mathcal{N}^{(m)}_A \le
  \frac{\abs{A}^2}{4\pi^2}.
  \]
\end{thm}

\begin{proof}
  By Proposition \ref{T:kernels} and a general formula for the
  variance of the counting function of a determinantal point process
  (see \cite[Appendix B]{Gu}),
  \[
  \var \mathcal{N}_A - \var \mathcal{N}_A^{(m)} =
  \frac{1}{4\pi^2} \int_A \int_A \sin^2
  \left(\frac{n(x-y)}{2}\right)
  \left(\frac{1}{\sin^2\left(\frac{x-y}{2}\right)} - \frac{1}{m^2 \sin^2
      \left( \frac{x-y}{2m}\right)}\right) \ dx \ dy.
  \]
   As in the proof of Theorem
  \ref{T:TV-bound}, for $0 < \varphi \le \frac{\pi}{2}$,
  \begin{equation*}
      0 \le \frac{1}{\sin^2 \varphi} - \frac{1}{m^2 \sin^2
          \left(\frac{\varphi}{m}\right)} 
        \le \frac{1}{\left(\varphi - \frac{1}{6} \varphi^3\right)^2}
        - \frac{1}{\varphi^2} \\
        \le 1,
  \end{equation*}
  from which the result follows.
\end{proof}

\section{Comparison of joint intensities}\label{S:correlations}

We conclude with the surprising fact that the joint intensities of the
process $\chi^{(m)}$ are always larger than those of the eigenvalue
process $\chi$; the implications of this observation remain mysterious
(at least to us).

\begin{prop}
  \label{T:intensities-inequality}
  For each $m$, $n$, and $k$, let $\rho_k:[0,2\pi)^k \to \R$
  denote the $k^{\mathrm{th}}$ joint intensity of the determinantal
  point process $\chi$, and let $\rho_k^{(m)}$ denote the
  $k^{\mathrm{th}}$ joint intensity of the determinantal point process
  $\chi^{(m)}$.  Then for each $x_1, \dots, x_k \in [0,2\pi)$,
  \[
  \rho_k^{(m)} (x_1, \dots, x_k) \ge \rho_k (x_1, \dots, x_k).
  \]
\end{prop}

\begin{proof}
  For this proof we use a different kernel which also generates the
  point process $\chi$ (see \cite{Dyson} or \cite[Section 5.2]{KaSa})
  :
  \[
  T_n(x,y) = \sum_{j=0}^{n-1} e^{ij(x-y)},
  \]
  which, by the same change of variables used in the proof of
  Proposition \ref{T:kernels}, implies that $\chi^{(m)}$ is generated
  by the kernel
  \begin{equation*}
    \begin{split}
      T_n^{(m)}(x,y) & = \frac{1}{m} \sum_{j=0}^{mn-1} e^{ij(x-y)/m} \\
      & = \frac{1}{m} \sum_{p=0}^{m-1} \sum_{q=0}^{n-1}
      e^{i(mq+p)(x-y)/m} \\
      & = \frac{1}{m} \sum_{p=0}^{m-1} e^{ipx/m} T_n(x,y) e^{-ipy/m}.
    \end{split}
  \end{equation*}
  It follows that
  \[
  [T_n^{(m)}(x_j,x_\ell)]_{j,\ell=1}^k = \frac{1}{m} \sum_{p=0}^{m-1} D^p
  [T_n(x_j,x_\ell)] (D^p)^*,
  \]
  where $D = \operatorname{diag} (e^{ix_1/m}, \dots, e^{ix_k/m})$ is a
  diagonal unitary matrix, and so by Minkowski's determinant
  inequality \cite[Corollary II.3.21]{Bhatia},
  \begin{equation*}
    \begin{split}
      \bigl(\rho^{(m)}_k(x_1, \dots, x_k)\bigr)^{1/k}
      & = \left(\det \left[ \frac{1}{m} \sum_{p=0}^{m-1} D^p
        [T_n(x_j,x_\ell)] (D^p)^* \right] \right)^{1/k} \\
      & \ge \frac{1}{m} \sum_{p=0}^{m-1} \left(\det \bigl(  D^p
        [T_n(x_j,x_\ell)] (D^p)^*\bigr) \right)^{1/k} \\
      & = \frac{1}{m} \sum_{p=0}^{m-1} \left(\det
        [T_n(x_j,x_\ell)] \right)^{1/k} \\
      & = \rho_k(x_1, \dots, x_k)^{1/k}.
      \qedhere
    \end{split}
  \end{equation*}
\end{proof}

\section*{Acknowledgements}

This research was partially supported by grants from the U.S.\
National Science Foundation (DMS-1308725 to E.M.) and the Simons
Foundation (\#315593 to M.M.).  This work was partly carried out while
the authors were visiting the Institut de Math\'ematiques de Toulouse
at the Universit\'e Paul Sabatier; the authors thank them for their
generous hospitality.

\bibliographystyle{amsplain}
\bibliography{unitary-similarity}

\providecommand{\bysame}{\leavevmode\hbox to3em{\hrulefill}\thinspace}
\providecommand{\MR}{\relax\ifhmode\unskip\space\fi MR }
\providecommand{\MRhref}[2]{%
  \href{http://www.ams.org/mathscinet-getitem?mr=#1}{#2}
}
\providecommand{\href}[2]{#2}
\begin{thebibliography}{10}

\bibitem{AGZ}
G.~W. Anderson, A.~Guionnet, and O.~Zeitouni, \emph{An introduction to random
  matrices}, Cambridge Studies in Advanced Mathematics, vol. 118, Cambridge
  University Press, Cambridge, 2010. \MR{2760897 (2011m:60016)}

\bibitem{Bhatia}
R.~Bhatia, \emph{Matrix analysis}, Graduate Texts in Mathematics, vol. 169,
  Springer-Verlag, New York, 1997. \MR{1477662 (98i:15003)}

\bibitem{CD03}
M.~Coram and P.~Diaconis, \emph{New tests of the correspondence between unitary
  eigenvalues and the zeros of {R}iemann's zeta function}, J. Phys. A
  \textbf{36} (2003), no.~12, 2883--2906, Random matrix theory. \MR{1986397
  (2004j:11098)}

\bibitem{CoLe}
O.~Costin and J.~L. Lebowitz, \emph{Gaussian fluctuation in random matrices},
  Phys. Rev. Lett. \textbf{75} (1995), no.~1, 69--72. \MR{3155254}

\bibitem{Dyson}
F.~J. Dyson, \emph{Correlations between eigenvalues of a random matrix}, Comm.
  Math. Phys. \textbf{19} (1970), 235--250. \MR{0278668 (43 \#4398)}

\bibitem{Feller2}
W.~Feller, \emph{An introduction to probability theory and its applications.
  {V}ol. {II}.}, second ed., John Wiley \& Sons, Inc., New York-London-Sydney,
  1971. \MR{0270403 (42 \#5292)}

\bibitem{Gu}
J.~Gustavsson, \emph{Gaussian fluctuations of eigenvalues in the {GUE}}, Ann.
  Inst. H. Poincar\'e Probab. Statist. \textbf{41} (2005), no.~2, 151--178.
  \MR{2124079 (2005k:60074)}

\bibitem{HKPV06}
J.~B. Hough, M.~Krishnapur, Y.~Peres, and B.~Vir{\'a}g, \emph{Determinantal
  processes and independence}, Probab. Surv. \textbf{3} (2006), 206--229.
  \MR{2216966 (2006m:60068)}

\bibitem{Kato95}
T.~Kato, \emph{Perturbation theory for linear operators}, Classics in
  Mathematics, Springer-Verlag, Berlin, 1995, Reprint of the 1980 edition.
  \MR{1335452 (96a:47025)}

\bibitem{KaSa}
N.~M. Katz and P.~Sarnak, \emph{Random matrices, {F}robenius eigenvalues, and
  monodromy}, American Mathematical Society Colloquium Publications, vol.~45,
  American Mathematical Society, Providence, RI, 1999. \MR{1659828
  (2000b:11070)}

\bibitem{Macchi}
O.~Macchi, \emph{The coincidence approach to stochastic point processes},
  Advances in Appl. Probability \textbf{7} (1975), 83--122. \MR{0380979 (52
  \#1876)}

\bibitem{MM-powers}
E.~S. Meckes and M.~W. Meckes, \emph{Spectral measures of powers of random
  matrices}, Electron. Commun. Probab. \textbf{18} (2013), no.~78, 1--13.
  \MR{3109633}

\bibitem{Mehta}
M.~L. Mehta, \emph{Random matrices}, third ed., Pure and Applied Mathematics
  (Amsterdam), vol. 142, Elsevier/Academic Press, Amsterdam, 2004. \MR{2129906
  (2006b:82001)}

\bibitem{Ra97}
E.~M. Rains, \emph{High powers of random elements of compact {L}ie groups},
  Probab. Theory Related Fields \textbf{107} (1997), no.~2, 219--241.
  \MR{1431220 (98b:15026)}

\bibitem{Ra03}
\bysame, \emph{Images of eigenvalue distributions under power maps}, Probab.
  Theory Related Fields \textbf{125} (2003), no.~4, 522--538. \MR{1974413
  (2004e:15029)}

\bibitem{Soshnikov}
A.~Soshnikov, \emph{Determinantal random point fields}, Uspekhi Mat. Nauk
  \textbf{55} (2000), no.~5(335), 107--160. \MR{1799012 (2002f:60097)}

\bibitem{So00}
A.~B. Soshnikov, \emph{Gaussian fluctuation for the number of particles in
  {A}iry, {B}essel, sine, and other determinantal random point fields}, J.
  Statist. Phys. \textbf{100} (2000), no.~3-4, 491--522. \MR{1788476
  (2001m:82006)}

\bibitem{Villani}
C.~Villani, \emph{Optimal transport, old and new}, Grundlehren der
  Mathematischen Wissenschaften [Fundamental Principles of Mathematical
  Sciences], vol. 338, Springer-Verlag, Berlin, 2009. \MR{2459454
  (2010f:49001)}

\bibitem{Wieand}
K.~Wieand, \emph{Eigenvalue distributions of random unitary matrices}, Probab.
  Theory Related Fields \textbf{123} (2002), no.~2, 202--224. \MR{1900322
  (2003b:60016)}

\bibitem{Wojtaszczyk}
P.~Wojtaszczyk, \emph{Banach spaces for analysts}, Cambridge Studies in
  Advanced Mathematics, vol.~25, Cambridge University Press, Cambridge, 1991.
  \MR{1144277 (93d:46001)}

\end{thebibliography}

%%% AUTHOR: Include a short description of each author following the
%%% structure below. Use the same short tags used previously.  
%%% Use \imageat{} and \imagedot{} instead of "@" and "." in
%%% email addresses-this replaces the symbols with graphics to avoid 
%%% e-mail address harvesting from the .pdf file
\begin{dajauthors}
\begin{authorinfo}[elizabeth]
  Elizabeth S.\ Meckes\\
  Case Western Reserve University\\
  Cleveland, Ohio, USA\\
  elizabeth.meckes\imageat{}case\imagedot{}edu \\
  \url{https://www.case.edu/artsci/math/esmeckes/}
\end{authorinfo}
\begin{authorinfo}[mark]
  Mark W.\ Meckes\\
  Case Western Reserve University\\
  Cleveland, Ohio, USA\\
  mark.meckes\imageat{}case\imagedot{}edu \\
  \url{https://www.case.edu/artsci/math/mwmeckes/}
\end{authorinfo}
%\begin{authorinfo}[laci]
%  L\'aszl\'o Lov\'asz\\
%  Professor\\
%  E\"otv\"os Lor\'and University\\
%  Budapest, Hungary\\
% laci\imageat{}comb\imagedot{}elte\imagedot{}hu\\
%  \url{http://www.cs.elte.hu/~lovasz}
%\end{authorinfo}
%\begin{authorinfo}[andy]
%  Andrew Chi-Chih Yao\\
%  Professor\\
%  etc.
%\end{authorinfo}
\end{dajauthors}

\end{document}